\documentclass[11pt,letterpaper]{amsart}
\usepackage{joshuag}
\usepackage{joshuag-layout}
\usepackage{comment}
\usepackage{graphs}

\newcommand{\mfsl}{\mathfrak{sl}}

\DeclareMathOperator{\GL}{GL}

\DeclareMathOperator{\tr}{tr}
\DeclareMathOperator{\Span}{Span}
\DeclareMathOperator{\Out}{Out}
\DeclareMathOperator{\Inn}{Inn}
\DeclareMathOperator{\Der}{Der}
\newcommand{\outerAut}[1]{^{#1}}

\renewcommand{\lang}[1]{{\sc #1}}

\newcommand{\LAC}{\lang{LAC}}
\newcommand{\subsp}[1]{\mathcal{#1}}
\newcommand{\ad}{\ensuremath{\text{ad}}}

\newcommand{\mylang}{Problem A}

\newcommand{\problem}[3]{
\begin{quotation}
\textbf{Problem: } #1

\noindent \textbf{Input: } #2

\noindent \textbf{Output: } #3
\end{quotation}
}

\newtheorem*{boundedAbThm}{Theorem~\ref{thm:boundedAb}}
\newtheorem*{boundedSSThm}{Theorem~\ref{thm:boundedSemisimple}}
\newtheorem*{abThm}{Theorem~\ref{thm:abGI}}
\newtheorem*{ssThm}{Theorem~\ref{thm:ssGI}}
\newtheorem*{boundedReducibleThm}{Theorem~\ref{thm:boundedReducible}}

\begin{document}

\title{Lie algebra conjugacy\footnote{Research partially supported by NSF grants DMS-0652521 and CCF-1017760.}}
\author{Joshua A. Grochow \\
University of Chicago \\
joshuag@cs.uchicago.edu}

\begin{abstract}
\normalsize We study the problem of matrix Lie algebra conjugacy. Lie algebras arise centrally in areas as diverse as differential equations, particle physics, group theory, and the Mulmuley--Sohoni Geometric Complexity Theory program. A matrix Lie algebra is a set $\subsp{L}$ of matrices such that $M_{1}, M_{2} \in \subsp{L} \implies M_{1} M_{2} - M_{2} M_{1} \in \subsp{L}$. Two matrix Lie algebras are conjugate if there is an invertible matrix $M$ such that $\subsp{L}_{1} = M \subsp{L}_{2} M^{-1}$. 

We show that certain cases of Lie algebra conjugacy are equivalent to graph isomorphism. On the other hand, we give polynomial-time algorithms for other cases of Lie algebra conjugacy, which allow us to essentially derandomize a recent result of Kayal on affine equivalence of polynomials. Affine equivalence is related to many complexity problems such as factoring integers, graph isomorphism, matrix multiplication, and permanent versus determinant.

Specifically, we show:
\begin{itemize}
\item Abelian Lie algebra conjugacy is equivalent to the code equivalence problem, and hence is as hard as graph isomorphism. A Lie algebra is abelian if all of its matrices commute pairwise.

\item Abelian Lie algebra conjugacy of $n \times n$ matrices can be solved in $poly(n)$ time when the Lie algebras have dimension $O(1)$. The dimension of a Lie algebra is the maximum number of linearly independent matrices it contains.

\item Semisimple Lie algebra conjugacy is equivalent to graph isomorphism. A Lie algebra is semisimple if it is a direct sum of simple Lie algebras.

\item Semisimple Lie algebra conjugacy of $n \times n$ matrices can be solved in polynomial time when the Lie algebras consist of only $O(\log n)$ simple direct summands.

\item Conjugacy of completely reducible Lie algebras---that is, a direct sum of an abelian and a semisimple Lie algebra---can be solved in polynomial time when the abelian part has dimension $O(1)$ and the semisimple part has $O(\log n)$ simple direct summands.
\end{itemize}
\end{abstract}

\maketitle

\pagestyle{myheadings}
\markboth{Lie Algebra Conjugacy---Joshua A. Grochow}{Lie Algebra Conjugacy---Joshua A. Grochow}


\newcommand{\var}[1]{\mathbf{#1}}
\section{Introduction}
A \definedWord{matrix Lie algebra} is defined as a set of $n \times n$ matrices closed under the following operations: multiplication by scalars $A \mapsto \alpha A$ for $\alpha \in \C$, the usual matrix addition, and a multiplication-like operation denoted $[A, B] := AB - BA$. Lie algebras are an important tool in areas as diverse as differential equations \cite{olver, steeb}, particle physics \cite{georgi}, group theory \cite{fultonHarris, chevalley, vinberg}, and the Mulmuley--Sohoni Geometric Complexity Theory program \cite{gct1}. 

In complexity theory, Kayal \cite{kayal} has recently used Lie algebras in the so-called affine equivalence problem, which arises in many areas of complexity: factoring integers, permanent versus determinant, matrix multiplication, lower bounds for depth-three circuits, and several more (see \cite[\S 1.1]{kayal}). Kayal essentially used Lie algebra conjugacy to give a randomized polynomial-time algorithm to decide when a function can be gotten from the determinant by an invertible linear change of variables. This is the affine equivalence problem for the determinant.

The following are examples of Lie algebras, which should help give their flavor, and introduces some of those Lie algebras on which we prove results, namely abelian, diagonalizable, and \linebreak (semi-)simple:
\begin{enumerate}
\item The collection of all $n \times n$ matrices.

\item \label{ex:abelian} The collection of all diagonal $n \times n$ matrices is a Lie algebra of dimension $n$. Any two diagonal matrices $D_{1}, D_{2}$ commute. Since $D_{1} D_{2} - D_{2} D_{1} = 0$ this is a Lie algebra. Any Lie algebra in which all matrices commute is called \definedWord{abelian}.

\item \label{ex:diag} In fact, \emph{any} collection of diagonal matrices that is closed under taking linear combinations is a Lie algebra, for the same reason as above. Furthermore, if $\subsp{D}$ is such a Lie algebra, then $A \subsp{D} A^{-1}$ is as well, since conjugating by $A$ preserves the fact that all the matrices in $\subsp{D}$ commute. Any Lie algebra conjugate to a set of diagonal matrices is called \definedWord{diagonalizable}.


\item \label{ex:simple} The collection of all $n \times n$ matrices with trace zero. Since $\tr(AB - BA) = 0$ for any $A, B$, this is also a Lie algebra. This is an example of a \definedWord{simple} Lie algebra.

\item \label{ex:kayal} The collection of all $2n \times 2n$ matrices of the form $\left(\begin{array}{cc} C & 0 \\ 0 & D \end{array}\right)$ where $C, D$ are $n \times n$ matrices and $\tr C + \tr D  = 0$. 
\end{enumerate}

Two Lie algebra $\subsp{L}_{1}, \subsp{L}_{2}$ are \definedWord{conjugate} if there is an invertible matrix $A$ such that $\subsp{L}_{1} = A \subsp{L}_{2} A^{-1}$. Although it was not phrased this way, Kayal gave a randomized reduction from the affine equivalence  problem for the determinant to the Lie algebra conjugacy problem for Lie algebras isomorphic to example (\ref{ex:kayal}). However, where he uses properties very specific to the Lie algebras associated to permanent and determinant 
that can be computed using randomization, we are able to instead use a deterministic approach to the more general problem of Lie algebra conjugacy.

\subsection{Results}
We show that certain cases of Lie algebra conjugacy are solvable in polynomial time. We also show that extending these cases is difficult, as such an extension is equivalent to graph isomorphism in one case and at least as hard as graph isomorphism in the other case. One of these cases is strong enough to mostly derandomize Kayal's result \cite{kayal} on testing affine equivalence to the determinant (see \S \ref{sec:affine} for details). 

We now give the formal definition of Lie algebra conjugacy. Since Lie algebras are closed under taking linear combinations, we can give them as input to algorithms by providing a linear basis.
\problem{Lie Algebra Conjugacy (\lang{LAC})}
{Two Lie algebras $\subsp{L}_{1}$ and $\subsp{L}_{2}$ of $n \times n$ matrices, given by basis elements.}
{An invertible $n \times n $ matrix $A$ such that $A \subsp{L}_{1} A^{-1} = \subsp{L}_{2}$, if such $A$ exists, otherwise ``the Lie algebras are not conjugate.''}

Recall the definitions of abelian and diagonalizable from examples (\ref{ex:abelian}) and (\ref{ex:diag}) above, respectively. 
\begin{boundedAbThm} Abelian diagonalizable Lie algebra conjugacy of $n \times n$ matrices can be solved in $poly(n)$ time when the Lie algebras have dimension $O(1)$. \end{boundedAbThm}

Abelian Lie algebras are one of two fundamental building blocks of all Lie algebras. The other fundamental building blocks are the semisimple Lie algebras. A Lie algebra is semisimple if it is a direct sum of simple Lie algebras. Example (\ref{ex:simple}) above is a simple Lie algebra; see Appendix~\ref{sec:background} for the full definition, and the discussion leading up to Remark~\ref{rmk:buildingBlocks} for what we mean by ``building blocks.'' For this other building block, we show a similar result:

\begin{boundedSSThm} Semisimple Lie algebra conjugacy of $n \times n$ matrices can be solved in $poly(n)$ time when the Lie algebras have only $O(\log n)$ simple direct summands. \end{boundedSSThm}

Despite the $O(\log n)$ restriction, Theorem~\ref{thm:boundedSemisimple} is already strong enough to mostly derandomize Kayal's result (Corollary~\ref{cor:kayal} below). Also, note that even a single simple Lie algebra can have unbounded dimension, as in example~(\ref{ex:simple}), let alone a semisimple one with $O(\log n)$ simple summands. Theorem~\ref{thm:boundedIrreps} gives another class on which Lie algebra conjugacy is solvable in polynomial time.

For both results above, we show that removing the quantitative restrictions is likely to be difficult:

\begin{abThm} Graph isomorphism polynomial-time reduces to abelian diagonalizable Lie algebra conjugacy, when the Lie algebras may have unbounded dimension. \end{abThm}

\begin{ssThm} Graph isomorphism is equivalent to semisimple Lie algebra conjugacy, when the Lie algebras may contain an unbounded number of simple direct summands. \end{ssThm}

In fact, we show that abelian diagonaliable Lie algebra conjugacy is equivalent to the code equivalence problem. The code equivalence problem is to test whether two subspaces of a vector space can be made equal by permuting their coordinates. For codes over $\F_{2}$, code equivalence was known to be as hard as graph isomorphism \cite{petrankRoth}; we extend their proof to show that code equivalence over \emph{any field} is as hard as graph isomorphism.

Finally, we combine the abelian and semisimple cases together, to show results on Lie algebra conjugacy when the Lie algebras are the direct sum of an abelian Lie algebra and a semisimple one:

\begin{boundedReducibleThm}
Conjugacy of Lie algebras of $n \times n$ matrices can be determined in $poly(n)$ time when the Lie algebras are a direct sum of an $O(1)$-dimensional abelian diagonalizable Lie algebra and a semisimple Lie algebra with $O(\log n)$ simple direct summands.
\end{boundedReducibleThm}

Since abelian is a special case of abelian-plus-semisimple, this more general case is obviously as hard as code equivalence when we drop the quantitative restrictions of the above theorem.

\subsection{Outline}
In \S\ref{sec:diag} we prove Theorems~\ref{thm:boundedAb} and \ref{thm:abGI} on abelian \LAC; this section can be understood without any background on Lie algebras. Our other results require more knowledge of Lie algebras; we collect the necessary background in Appendix~\ref{sec:background}. In \S\ref{sec:semisimple} we prove Theorems~\ref{thm:ssGI}, \ref{thm:boundedSemisimple}, and \ref{thm:boundedIrreps} on semisimple \LAC. In \S\ref{sec:reducible} we prove our results on direct sums of abelian and semisimple Lie algebras, including Theorem~\ref{thm:boundedReducible}. In \S\ref{sec:affine} we show how to use the above machinery to essentially derandomize Kayal's result on testing affine equivalence to the determinant. 
In the final section, we discuss how close the abelian-plus-semisimple case is to the general case, directions toward the general case, and other future work, 
including potential ways to solve important special cases of the affine equivalence problem efficiently without having to efficiently solve \lang{GI}.

\section{Abelian diagonalizable Lie algebra conjugacy and code equivalence} \label{sec:diag}
In this section we show that conjugacy of abelian diagonalizable matrix Lie algebras is Karp-equivalent to the code equivalence problem, and hence is at least as hard as \lang{GI}. The reduction to code equivalence allows us to solve abelian diagonalizable \lang{LAC} for constant-dimensional Lie algebras in polynomial time. 

A $d$-dimensional \definedWord{code} of length $n$ over a field $\F$ is a $d$-dimensional subspace of $\F^{n}$. Codes are represented algorithmically by giving bases for them as subspaces. The symmetric group $S_{n}$ acts on $\F^{n}$ by permutation of coordinates: for $\pi \in S_{n}$ and $\vec{\alpha} = (\alpha_{1}, \dotsc, \alpha_{n}) \in \F^{n}$, $\pi \cdot \vec{\alpha} = (\alpha_{\pi(1)}, \dotsc, \alpha_{\pi(n)})$. $S_{n}$ then acts on a subspace $V \subseteq \F^{n}$ by $\pi \cdot V = \set{\pi \cdot v : v \in V}$. The \definedWord{code equivalence problem} is: given two codes $C_{1}, C_{2}$, determine whether there is a permutation $\pi \in S_{n}$ such that $\pi \cdot C_{1} = C_{2}$.

\begin{thm} \label{thm:abGI} Abelian dagonalizable Lie algebra conjugacy for $d$-dimensional subspaces of $n \times n$ matrices is equivalent to code equivalence for $d$-dimensional subspaces of $\F^{n}$, for any field $\F$. \end{thm}

Before proving this theorem we give some of its consequences.

\begin{cor} \label{thm:boundedAb} 
Conjugacy of abelian diagonalizable $O(1)$-dimensional Lie algebras of $n \times n$ matrices over any field can be solved in $poly(n)$ time.
\end{cor}

\begin{proof}
Babai (see \cite[Theorem~7.1]{BCGQ}) showed that, over any field $\F$, equivalence of $d$-dimensional linear codes of length $n$ reduces to $\binom{n}{d}$ instances of $d \times (n - d)$ edge-colored bipartite graph isomorphism. Each such instance can be solved in $poly(n) \cdot \min\{d!, (n-d)!\}$ time, so when $d = O(1)$ code equivalence can be solved in polynomial time. By Theorem~\ref{thm:abGI}, $d$-dimensional diagonalizable Lie algebra conjugacy can be solved in polynomial time when $d = O(1)$.
\end{proof}


\begin{cor} \label{cor:diagGIhard} Abelian diagonalizable Lie algebra conjugacy over any field is \lang{GI}-hard. \end{cor}

\begin{proof}
Petrank and Roth \cite{petrankRoth} showed that \lang{GI} Karp-reduces to code equivalence over $\F_{2}$. Over an arbitrary field we use the same reduction, but an extension of their proof is required, which we give in Lemma~\ref{lem:PR} below. Theorem~\ref{thm:abGI} then shows that abelian diagonalizable \lang{LAC} is \lang{GI}-hard.
\end{proof}

\begin{lem} \label{lem:PR} Graph isomorphism Karp-reduces to code equivalence over any field $\F$. \end{lem}

\begin{proof}
Given a graph $G$, we construct the generator matrix for a code over $\F$ such that two graphs are isomorphic if and only if the codes are equivalent. Let $M(G) = [ I_{m} | I_{m} | I_{m} | D ]$ where $m = |E(G)|$, $I_{m}$ is the $m \times m$ identity matrix, and $D$ is the incidence matrix of $G$: 
\[
D_{e,v} = \begin{cases}
1 & \If v \in e \\
0 & \otherwise
\end{cases}
\]
The \definedWord{Hamming weight} of a vector over $\F$ is the number of non-zero entries. The following claim is essentially the crux of Petrank and Roth's argument, but generalized so as to apply over any field.

\textit{Claim: } up to permutation and scaling of the rows, $M(G)$ is the unique generator matrix of its code which satisfies the following properties:
\begin{enumerate}
\item it is a $|E| \times (3|E| + |V|)$ generator matrix;
\item each row has Hamming weight $\leq 5$;
\item any nondegenerate linear combination of two or more rows has Hamming weight $\geq 6$
\end{enumerate}
A linear combination of $k$ rows is nondegenerate if all $k$ of its coefficients are nonzero.

\textit{Proof of claim: } First, $M(G)$ satisfies (1)--(3). The only part to check is (3): in the first $3m$ columns, any nondegenerate linear combination of $k \geq 2$ rows will have $3k \geq 6$ nonzero entries. Next, let $C$ denote the code generated by the rows of $M(G)$. By (2) and (3) the rows of $M(G)$ are the \emph{unique} vectors in $C$ (up to scaling) of Hamming weight $\leq 5$. Hence if $M'$ is any other generator matrix of $C$ satisfying (1)--(3), its rows must be scaled versions of the rows of $M(G)$ in some order. This proves the claim.

Now, suppose that $M(G_{1})$ and $M(G_{2})$ generate equivalent codes. Then there is a nonsingular matrix $S$ and a permutation matrix $P$ such that $M(G_{1}) = SM(G_{2})P$. By the claim, $S=\Delta S'$ where $\Delta$ is diagonal and $S'$ is a permutation matrix. However, since the first $3|E|$ columns of $M(G_{1})$ and $M(G_{2})$ only contain $0,1$-entries, $\Delta = I$. The rest of the proof of the reduction, including the other direction, proceeds exactly as in Petrank and Roth \cite{petrankRoth}. \end{proof}

\begin{proof}[Proof of Theorem~\ref{thm:abGI}]
Let $(A_{1}, \dotsc, A_{d}), (B_{1}, \dotsc, B_{d})$ be an instance of abelian diagonalizable \lang{LAC}. Standard techniques in linear algebra can be used to simultaneously diagonalize the $A_{i}$ in polynomial time, so we may now assume that the $A_{i}$ are in fact diagonal, rather than merely diagonalizable. Similarly for the $B_{i}$. Let $\subsp{A}$, resp. $\subsp{B}$, denote the Lie algebras spanned by the $A_{i}$, resp. $B_{i}$.

\textit{Claim: } If $\subsp{A}$ and $\subsp{B}$ are diagonal, then they are conjugate if and only if they are conjugate by a permutation matrix. 

By ``flattening out'' the entries of the diagonal matrices into ``row'' vectors the claim shows that diagonalizable $d$-dimensional \lang{LAC} of $n \times n$ matrices is Karp-equivalent to $d$-dimensional code equivalence of codes of length $n$. Thus the claim will complete the proof of the theorem.

\textit{Proof of claim: } Suppose $g \subsp{A} g^{-1} = \subsp{B}$. Since $\subsp{B}$ is diagonal, $g$ must preserve the eigenspaces of every matrix in $\subsp{A}$. The formalization of this notion will allow us to prove our claim. Let $\lambda_{i} \colon \subsp{A} \to \F$ be the linear function $\lambda_{i}(A) = A_{ii}$. We can think of $\lambda_{i}$ as a ``simultaneous eigenvalue for the space $\subsp{A}$ of matrices,'' generalizing the notion of an eigenvalue of a single matrix. Such functions are called \definedWord{weights} in the theory of Lie algebras, and they will play a significant role here and in the case of semisimple Lie algebras as well. Analogous to an eigenspace corresponding to an eigenvalue, there are \definedWord{weight spaces} corresponding to weights. Namely, if $\lambda\colon \subsp{A} \to \F$ is a weight, the corresponding weight space is
\[
V_{\lambda}(\subsp{A}) := \{v \in \F^{n} : Av = \lambda(A) v \text{ for all } A \in \subsp{A}\}
\]
It is these weight spaces that $g$ must preserve in order for $g \subsp{A} g^{-1}$ to be diagonal. For example, if every weight space is $1$-dimensional---or equivalently, if for every pair of indices $1 \leq i < j \leq n$ there is some matrix $A \in \subsp{A}$ with $A_{ii} \neq A_{jj}$---then $g$ must be a permutation matrix. 

More generally, $g$ may send $v \in V_{\lambda_{1}}$ into $V_{\lambda_{2}}$ if and only if $g V_{\lambda_{1}} = V_{\lambda_{2}}$. Within each weight space, $g$ may act in an arbitrary invertible manner. In other words, $g$ is composed of invertible blocks of dimension $\dim V_{\lambda_{i}}$, the pattern in which these blocks appear is a permutation, and that permutation may send $i \mapsto j$ if and only if $\dim V_{\lambda_{i}} = \dim V_{\lambda_{j}}$. However, if $g'$ has the same permutation pattern as $g$ but all the blocks in $g'$ are the identity, then $g \subsp{A} g^{-1} = g' \subsp{A} g'^{-1}$. Hence, without loss of generality, we may take $g$ to be a permutation matrix, proving the claim.
\end{proof}

\section{Semisimple Lie algebra conjugacy and graph isomorphism} \label{sec:semisimple}
\begin{thm} \label{thm:ssGI} Semisimple Lie algebra conjugacy is equivalent to graph isomorphism.
\end{thm}

\begin{proof}
We break the proof into four lemmas. By Lemma~\ref{lem:outer}, semisimple Lie algebra conjugacy is equivalent to deciding whether two representations of a semisimple Lie algebra are equivalent up to outer automorphism. By Lemma~\ref{lem:a}, the latter problem reduces to a special case of twisted code equivalence with multiplicities, which we refer to as \mylang. Finally, Lemma~\ref{lem:aGI} reduces \mylang\ to graph isomorphism, and Lemma~\ref{lem:GIa} reduces graph isomorphism to semisimple Lie algebra conjugacy.
\end{proof}

\begin{lem}[de Graaf\footnote{This lemma is essentially present in de Graaf's book \cite{deGraafBook}, especially the content leading up to the discussion at the end of his Section~8.5. However, de Graaf's discussion is presented in terms of weights and the choice of Cartan subalgebra, whereas the aspect we wish to highlight requires no mention of these topics, and can be explained by completely elementary means.}] \label{lem:outer} 
Semisimple Lie algebra conjugacy is equivalent to---nearly just a restatement of---the following problem (see Appendix~\ref{sec:background:rep} for definitions):
\problem{Outer equivalence of Lie algebra representations}
{Two faithful representations $\rho_{1}, \rho_{2}\colon \subsp{L} \to M_{n}$ of a semisimple (abstract) Lie algebra $\subsp{L}$. The $\rho_{i}$ are given by the matrices $\rho_{i}(b_{j})$ for some basis $b_{1}, \dotsc, b_{d}$ of $\subsp{L}$, and $\subsp{L}$ is given by structural constants in the $b_{i}$ basis (see Appendix~\ref{sec:background}).}
{An outer automorphism $\alpha \in \Out(\subsp{L})$ such that $\rho_{1}\outerAut{\alpha}$ is equivalent to $\rho_{2}$, or ``the two representations are not equivalent up to automorphism.''}
\end{lem}

\begin{proof}
Suppose $\subsp{L}_{1}, \subsp{L}_{2} \subseteq M_{n}$ is an instance of semisimple Lie algebra conjugacy, that is, they are both semisimple matrix Lie algebras. Using techniques given in de Graaf \cite[\S 5.11]{deGraafBook}, we can determine if the $\subsp{L}_{i}$ are isomorphic as abstract Lie algebras; if not, they are not conjugate as matrix Lie algebras, or if so, we can construct an abstract Lie algebra $\subsp{L}$ that they are isomorphic to, together with isomorphisms $\rho_{i}\colon \subsp{L} \to \subsp{L}_{i}$ for $i=1,2$. Since $\subsp{L}_{i} \subseteq M_{n}$, the $\rho_{i}$ are faithful representations of $\subsp{L}$. We claim that the $\rho_{i}$ are equivalent up to an outer automorphism of $\subsp{L}$ if and only if the $\subsp{L}_{i}$ are conjugate.

Suppose $\subsp{L}_{2} = g \subsp{L}_{1} g^{-1}$ for some invertible matrix $g$. Let $c_{g}\colon M_{n} \to M_{n}$ be defined by $c_{g}(X) = gXg^{-1}$. Then $\alpha = \rho_{2}^{-1} \circ c_{g} \circ \rho_{1}$ is a map from $\subsp{L}$ to $\subsp{L}$. Since the $\rho_{i}$ are isomorphisms, and $c_{g}|_{\subsp{L}_{1}}\colon \subsp{L}_{1} \to \subsp{L}_{2}$ is an isomorphism, the composition $\alpha$ is an automorphism of $\subsp{L}$. Then $\rho_{2} \circ \alpha = \rho_{2} \circ \rho_{2}^{-1} \circ c_{g} \circ \rho_{1} = c_{g} \circ \rho_{1}$, which is by definition equivalent to $\rho_{1}$. By the discussion following Lemma~\ref{lem:inner}, $\rho_{2}^{\overline{\alpha}}$ is thus equivalent to $\rho_{1}$, where $\overline{\alpha}$ is the outer automorphism corresponding to $\alpha$.

Conversely, suppose $\rho_{1}^{\overline{\alpha}}$ is equivalent to $\rho_{2}$ for some outer automorphism $\overline{\alpha}$. Let $\alpha \in \Aut(\subsp{L})$ be a representative of $\overline{\alpha}$; then there is an invertible matrix $g$ such that $\rho_{2} = c_{g} \circ \rho_{1} \circ \alpha$. Then we have 
\[
\subsp{L}_{2} = \im(\rho_{2}) =\im(c_{g}  \circ \rho_{1} \circ \alpha) = c_{g}(\im(\rho_{1} \circ \alpha)).
\]
Since $\alpha$ is an automorphism it is onto, so $\im(\rho_{1} \circ \alpha) = \im(\rho_{1}) = \subsp{L}_{1}$, and we have $\subsp{L}_{2} = g \subsp{L}_{1} g^{-1}$.

The preceding argument gives a reduction from semisimple Lie algebra conjugacy to the outer equivalence of Lie algebra representations. The reduction in the other direction is as follows: suppose $\subsp{L}$ is a semisimple Lie algebra and $\rho_{1},\rho_{2} \colon \subsp{L} \to M_{n}$ are two faithful representations. We reduce this to the instance of semisimple Lie algebra conjugacy given by $\subsp{L}_{i} = \im(\rho_{i})$ ($i=1,2$). The proof that this is a reduction is identical to the proof above.
\end{proof}

\begin{lem} \label{lem:a} Outer equivalence of Lie algebra representations reduces to the following problem:
\problem{\mylang}
{Two $r \times s$ integer matrices $M_{1}, M_{2}$; a partition of the columns into consecutive ranges $[1, \dotsc, k_{1}],[k_{1}+1,\dotsc, k_{1} + k_{2}], \dotsc [k_{1} + \dotsb + k_{t-1} + 1, \dotsc, s]$; for each range, a group $G_{\ell}$ acting on the integers appearing in the corresponding columns, where each $G_{\ell}$ is abstractly isomorphic to one of: $1$, $S_{2}$, or $S_{3}$.}
{A permutation $\pi \in S_{r}$, a permutation $\sigma \in S_{k_{1}} \times S_{k_{2}} \times \dotsb \times S_{k_{t}}$, and for each column an element $g_{j}$ in the group $G_{\ell}$ associated to that column range, such that for all $i,j$, $M_{1}(i,j) = g_{j}(M_{2}(\pi(i), \sigma(j)))$, or ``the matrices are not equivalent.'' In other words, after applying $\pi$ to the rows, $\sigma$ to the columns, and each $g_{j}$ to the values of the entries in the $j$-th column, $M_{1}$ and $M_{2}$ become equal.}
\end{lem}

\begin{proof}
Let $\subsp{L}$ be a semisimple Lie algebra, and let $\rho_{1}, \rho_{2} \colon \subsp{L} \to M_{n}$ be two faithful representations of $\subsp{L}$.
Compute the direct sum decomposition of $\subsp{L}$; suppose it is $\subsp{L} = \subsp{L}_{1,1} \oplus \dotsb \oplus \subsp{L}_{1,k_{1}} \oplus \subsp{L}_{2, 1} \oplus \dotsb \oplus \subsp{L}_{2,k_{2}} \oplus \dotsb \oplus \subsp{L}_{t,k_{t}}$ where each $\subsp{L}_{i,j}$ is a simple summand of $\subsp{L}$, and the $\subsp{L}_{i,j}$ are grouped by isomorphism type, so that $\subsp{L}_{i_{1},j_{1}}$ and $\subsp{L}_{i_{2},j_{2}}$ are isomorphic if and only if $i_{1} = i_{2}$. For each $i$, let $\subsp{L}_{i}$ be a simple Lie algebra isomorphic to $\subsp{L}_{i,j}$ for all $j$. 

To each $\rho_{i}$ we will associate a matrix $M_{i}$, as well as the other data necessary for \mylang. The columns correspond to the direct summands $\subsp{L}_{i,j}$, and the column partition is along the isomorphism types of the summands.

Next, we define the permutation groups $G_{\ell}$. To each simple type $\subsp{L}_{\ell}$, we fix once and for all an encoding of its representations as integers; both the encoding and decoding should be polynomial-time. That this can be done follows from the standard description of the representations of the simple Lie algebras. The integer $0$ will always stand for the (trivial) zero representation. The permutation action of $\Out(\subsp{L}_{\ell})$ on the representations of $\subsp{L}_{\ell}$, encoded as integers, can be easily computed, as follows. Given $\overline{\alpha} \in \Out(\subsp{L}_{\ell})$ and an integer, convert it to the corresponding representation as above. This representation is a linear map $\subsp{L}_{\ell} \to M_{n}$ for some $n$. Pre-compose this map with a representative $\alpha \in \Aut(\subsp{L}_{\ell})$ of $\overline{\alpha}$; this can be done because the outer automorphisms of all simple Lie algebras are known explicitly and are easy to compute. For example, the unique outer automorphisms of $\mfsl_{n}$, the trace zero matrices, is given by the map $A \mapsto -A^{T}$. The outer automorphism groups of simple Lie algebras are all trivial, $S_{2}$, or $S_{3}$. Finally, convert this new, ``twisted-by-$\alpha$'' representation back to an integer. The group $G_{\ell}$ associated to the $\ell$-th isomorphism type (=$\ell$-th column grouping) is then $\Out(\subsp{L}_{\ell})$, and the action on the integers is the action described above.

Finally, we describe the rows and the entries of the matrices $M_{i}$. Decompose the representations $\rho_{i}$ into their direct sum decompositions $\rho_{i} = \rho_{i,1} \oplus \dotsb \oplus \rho_{i,r}$, where each $\rho_{i,r}$ is an irreducible representation of $\subsp{L}$. Corollary~\ref{cor:directSum} says that this can be done in polynomial time. The $q$-th row of $M_{i}$ corresponds to the irreducible representations $\rho_{i,q}$. An irreducible representation of a direct sum of Lie algebras is completely specified by its restriction to each summand. Hence, the representation $\rho_{i,q}$ is specified by a representation of each summand $\subsp{L}_{\cdot, \cdot}$, that is, an integer in each column.

Since the outer automorphism group of $\subsp{L}$ is $\prod_{i=1}^{t} \Out(\subsp{L}_{i}) \wr S_{k_{i}} = \left(\prod_{i=1}^{t} \Out(\subsp{L}_{i})^{k_{i}}\right) \rtimes \left(\prod_{i=1}^{t} S_{k_{i}} \right)$, the representations $\rho_{1}, \rho_{2}$ are equivalent up to an outer automorphism if and only if there is a permutation of the columns (=direct summands of the Lie algebra), for each column an element $g_{\ell} \in G_{\ell}$ (=an outer automorphism of each direct summand), and a permutation of the rows (=irreducible constituents of $\rho_{i}$) that will make $M_{1}$ equal to $M_{2}$. Conversely, any such equivalence of $M_{1}$ and $M_{2}$ according to \mylang\ corresponds to an outer automorphism of $\subsp{L}$ that makes $\rho_{1}$ and $\rho_{2}$ equivalent.
\end{proof}

\begin{lem} \label{lem:aGI} \mylang\ reduces to graph isomorphism. \end{lem}

\begin{proof}
This is a fun exercise we invite the reader to try for him- or herself. We include the details in Appendix~\ref{app:aGI}.
\end{proof}

\begin{lem} \label{lem:GIa}
Graph isomorphism reduces to semisimple Lie algebra conjugacy.
\end{lem}

\begin{proof}
Let $(G_{1}, G_{2})$ be an instance of graph isomorphism, and let $D_{i}$ be the $0$-$1$ incidence matrix of $G_{i}$, where the rows correspond to edges and the columns correspond to vertices. The $G_{i}$ are isomorphic if and only if there is a permutation of the rows and the columns that makes the $D_{i}$ equal. Then $(D_{1}, D_{2})$ is an instance of \mylang\ where the column partition is trivial, and the column groups $G_{\ell}$ are also trivial. We show how to reduce such an instance of \mylang\ to outer equivalence of Lie algebra representations, and hence to semisimple Lie algebra conjugacy.

Given an instance of \mylang\ as above---in particular, it only contains the entries $0$ and $1$, it contains exactly two non-zero entries per row, every column contains a non-zero entry, and the column partition and column groups are all trivial---let $\subsp{L} = \mfsl_{2}^{\oplus n}$, where $n$ is the number of vertices of the $G_{i}$ (=columns of the matrices). Let a $1$ in the matrix $D_{i}$ correspond to the adjoint representation of $\mfsl_{2}$ (the Lie algebra of $2 \times 2$ trace zero matrices), which is faithful and has dimension $3$. Then, by reversing the reduction in Lemma~\ref{lem:a}, we get an instance of outer equivalence of Lie algebra representations.

Since each column contains a non-zero entry, these representations are faithful. Since each row contains exacty two $1$'s, the corresponding irreducible representation has dimension $3^{2} = 9$, hence the representations we get are matrices of dimension $9m \times 9m$, where $m$ is the number of edges of $G_{i}$. Since $\mfsl_{2}^{\oplus n}$ is generated by $3n$ elements, the representations can be specified by $3n \times (9m)^{2}$ numbers, which is polynomial in the size of the original graphs. 

Finally, although $\mfsl_{2}$ has an outer automorphism, this outer automorphism acts trivially on the representations of $\mfsl_{2}$, so the corresponding column groups are trivial, as desired.
\end{proof}

\begin{thm} \label{thm:boundedSemisimple} Conjugacy of semisimple Lie algebras of $n \times n$ matrices can be solved in polynomial time, when the Lie algebras consist of $O(\log n)$ simple direct summands. \end{thm}

\begin{proof}
If there are only $O(\log n / \log \log n)$ simple summands, then an elementary brute-force approach to \mylang\ works in $poly(n)$ time, since the number of outer automorphisms is $poly(n)$. However, when there are $O(\log n)$ simple summands, the number of outer automorphisms is $n^{O(\log n)}$, so we instead use a more sophisticated approach to twisted code equivalence, due to Babai, Codenotti, and Qiao \cite{BCQ} (cf. \cite[Theorem~4.2.1]{codenottiPhd}). \mylang\ is in fact a special case of twisted code equivalence with multiplicities, in which each row corresponds to a codeword. On the instance of \mylang\ corresponding to semisimple Lie algebras with $O(\log n)$ simple summands, their algorithm runs in $poly(n)$ time. Translating between their terminology and ours, the size of the code is the number of rows of the $M_{i}$, which is the number irreducible representations of the $\subsp{L}_{i}$, which is at most $n$, the size of the original matrices. Furthermore, the column groups $G_{i}$ all have bounded size. These two facts together imply that their algorithm runs in $poly(n)$ time.
\end{proof}

\begin{thm} \label{thm:boundedIrreps} Conjugacy of semisimple Lie algebras of $n \times n$ matrices can be solved in polynomial time, when the Lie algebras consist of $O(\log n / \log\log n)$ irreducible representations, and unboundedely many simple direct summands, at most $O(\log(n))$ of which have nontrivial outer automorphism actions on their representations. \end{thm}

In Appendix~\ref{sec:background:aut} we list the simple Lie algebras and their outer automorphism groups, and mention which have trivial actions on their representations. Three of the four infinite families of simple Lie algebras have this property, as well as four of the five exceptional simple Lie algebras.

\begin{proof}
In this case, the $M_{i}$ in the instance of \mylang\ have only $f(n) \leq O(\log n / \log\log n)$ rows. Although the size of the automorphism group may be more than polynomial, there are only polynomially many row permutations, so we only have to handle the outer automorphisms in each column exhaustively, and not the permutations between the columns. Specifically, try each combination of outer automorphisms of each column; since there are at most $O(\log n)$ columns with nontrivial outer automorphisms, and the outer automorphism group of a simple Lie algebra has size at most $6$, there are only $poly(n)$ possibilities. For each such possibility, try each of the $poly(n)$ many permutations of the rows, and for each check whether the set of columns of $M_{1}$ is equal to the set of columns of $M_{2}$.
\end{proof}

\section{Abelian plus semisimple (\ie, completely reducible)} \label{sec:reducible}
In this section, we describe how the algorithms and reductions for the abelian diagonalizable and semisimple cases fit into a single general framework and can be combined to handle the case of a direct sum of an abelian diagonalizable matrix Lie algebra with a semisimple matrix Lie algebra. This class of matrix Lie algebras is exactly the class of \definedWord{completely reducible} matrix Lie algebras. In the case of semisimple Lie algebras we used heavily the fact that all representations of semisimple Lie algebras can be written as a direct sum of irreducible representations (see Appendix~\ref{sec:background:rep}). The class we study in this section is the largest class of Lie algebras with this property (cf. Theorem~\ref{thm:CR}).

\begin{lem} \label{lem:outerAbCR} Lemma~\ref{lem:outer} applies to the class of abelian Lie algebras and the class of completely reducible matrix Lie algebras. \end{lem}

\begin{proof} The proof of Lemma~\ref{lem:outer} only required two ingredients: that the isomorphism problem for abstract Lie algebras of the class under consideration be efficiently solvable, and that twisting a representation by an inner automorphism leads to an equivalent representation. Both of these ingredients hold for  abelian Lie algebras: two abelian Lie algebras are abstractly isomorphic if and only if they have the same dimension, and abelian Lie algebras have no non-trivial inner automorphisms.

Similarly, a completely reducible matrix Lie algebra $\subsp{L}$ is a direct sum $\subsp{A} \oplus \subsp{S}$ where $\subsp{A}$ is abelian diagonalizable and $\subsp{S}$ is semisimple. The isomorphism problem for this class of Lie algebras is solvable in polynomial time. Finally, $\Inn(\subsp{L}) = \Inn(\subsp{A}) \times \Inn(\subsp{S}) \cong \Inn(\subsp{S})$ since abelian Lie algebras have no non-trivial inner automorphisms. Hence twisting a representation by an inner automorphism leads to an equivalent representation. \end{proof}

Although it was not originally phrased this way, we can now see that the algorithms and equivalences for abelian Lie algebra conjugacy in fact follow the same lines as those for semisimple Lie algebra conjugacy. The main difference is that the outer automorphism group of a $d$-dimensional abelian Lie algebra is the full general linear group $\GL_{d}$ of $d \times d$ invertible matrices---leading to linear code equivalence---whereas the outer automorphism group of a semisimple Lie algebra is close to $S_{n}$---leading to graph isomorphism. 

Furthermore, we can view Babai's reduction (see \cite[Theorem~7.1]{BCGQ}) from code equivalence as a sort of ``list normal form'' algorithm for the action of $\GL_{d}$ by automorphisms. Since $\GL_{d}$ acts by change of basis, we would like reduced row echelon form to be a normal form for this action. However, since one may permute the coordinates in the code equivalence problem, computing reduced row echelon form requires first picking the pivots. Babai's algorithm picks these pivots in all $\binom{n}{d}$ possible ways, reduces to row echelon form, and then uses graph isomorphism to handle the permutation action on the remaining coordinates of the code.

Combining these techniques yields:

\begin{thm} \label{thm:boundedReducible} Conjugacy of completely reducible matrix Lie algebras with an abelian diagonalizable part of dimension $a$, $s$ simple direct summands, and $r$ irreducible representation constituents reduces to $\binom{r}{a}$ instances of \mylang\ of size $r \times s$. In particular, completely reducible matrix Lie algebra conjugacy of $n \times n$ matrices can be solved in $poly(n)$ time under either of the following conditions:
\begin{itemize}
\item $a = O(\log n)$, $r = O(1)$, $s$ unbounded, and the number of simple summands with non-trivial outer automorphism action is at most $O(\log n)$; 

\item $a = O(1)$, $r$ unbounded, $s = O(\log n)$.
\end{itemize}
 \end{thm}

\section{Application to equivalence of polynomials} \label{sec:affine}
\begin{cor} \label{cor:kayal}
Given the Lie algebra of the symmetry group of a polynomial $f$ on $n^{2}$ variables, one can determine whether $f$ is linearly equivalent to $\det_{n}$ in \emph{deterministic} $poly(n)$ time.
\end{cor}

\begin{rmk} Computing the Lie algebra of the symmetry group of a polynomial is in fact \emph{equivalent} to polynomial identity testing, and hence cannot be derandomized without proving significant lower bounds \cite{kabanetsImpagliazzo}. Kayal \cite[Lemma~26]{kayal} shows how to compute the Lie algebra of the symmetry group of a polynomial given as a black-box, using the algorithm from \cite{kayal2} for computing the linear dependincies between a set of polynomials. Kayal \cite{kayal2} noted that computing such linear dependencies reduces to the search version of polynomial identity testing. The search and decision versions of polynomial identity testing are equivalent for low-degree functions. Conversely, a polynomial is constant if and only if its symmetry group consists of all invertible transformations of the variables. This holds if and only if the Lie algebra of its symmetry group consists of all linear transformations of the variables. Once this has been determined, evaluating the polynomial at any single point will determine whether it is zero or a non-zero constant.
\end{rmk}

In some sense, we have thus derandomized Kayal's algorithm as far as is possible in the black-box setting without derandomizing polynomial identity testing. Kayal uses randomization at several points, not just in the computation of the Lie algebra of the symmetry group, and we derandomize those using our deterministic algorithm for semisimple Lie algebra conjugacy from Theorem~\ref{thm:boundedSemisimple}. 

However, even in the dense, non-black-box setting this represents an improvement from $2^{O(n^{2})}$ to $2^{O(n \log n)}$. This is essentially optimal, since a generic function that is equivalent to $\det_{n}$ will include nearly all monomials of degree $n$ in $n^{2}$ variables, of which there are $2^{\Theta(n \log n)}$. By the dense setting we mean the setting in which $f$ is given by a list of coefficients of all monomials of degree $n$ in $n^{2}$ variables (without loss of generality, $f$ is homogeneous of degree $n$). Naive derandomization of Kayal's algorithm takes time $2^{O(n^{2})}$, since step (ii) of his \S 6.2.1 guesses a random element of a space of dimension $\Theta(n^{2})$. Similarly, testing affine equivalence to the determinant can be solved using quantifier elimination (see, \eg, Basu, Pollack, and Roy \cite[Ch.~14]{basuPollackRoy}), again in time $2^{O(n^{2})}$, because the witness to equivalence is an $n \times n$ matrix together with an $n$-dimensional vector. However, computing the Lie algebra of the symmetry group of $f$ only requires solving a linear system of $n^{2}$ equations in a number of variables equal to the number of monomials possible, which is roughly $\binom{n^{2}}{n} \leq n^{2n} = 2^{O(n \log n)}$. 

\begin{proof}[Proof of Corollary~\ref{cor:kayal}]
The Lie algebra of the symmetry group of $\det_{n}$ is $\mfsl_{n} \oplus \mfsl_{n}$, which has only two simple factors. By Theorem~\ref{thm:boundedSemisimple} we can test if the Lie algebra of the symmetry group of $f$ is conjugate to that of $\det_{n}$. If it is, then act on $f$ by the conjugating matrix so that the Lie algebra is now equal to that of $\det_{n}$. One might expect to then have to check whether $f(X) = f(X^{T})$, since this is also part of the symmetry group of $\det_{n}$, however, this is not necessary: in the case of the determinant, any function whose symmetry group has a Lie algebra conjugate to that of the determinant is in fact linearly equivalent to the determinant. Note that we have combined here all three main steps of Kayal's algorithm into a single reduction to Lie algebra conjugacy: Kayal uses the Lie algebra to reduce to permutational and scaling equivalence, then solves permutational equivalence and scaling equivalence separately.
\end{proof}


\section{Conclusion and future work} \label{sec:future}
Lie algebra conjugacy arises in Geometric Complexity Theory and the affine equivalence problem. We solved Lie algebra conjugacy over $\C$ in polynomial time for several important classes of Lie algebras---namely abelian, semisimple, and completely reducible (=abelian diagonalizable $\oplus$ semisimple)---under various quantitative constraints. We showed that without these quantitative constraints, these cases of Lie algebra conjugacy all become at least as hard as graph isomorphism.

The completely reducible case is not far from the general case, though significant obstacles remain. Levi's Theorem says that every Lie algebra is the semi-direct product of a solvable Lie algebra by a semisimple one; a solvable Lie algebra is an iterated extension of abelian Lie algebras (see Appendix~\ref{sec:background:morestruct} for definitions). The completely reducible case, which we resolved, restricts the solvable part to be abelian, and restricts the semidirect product to be direct. The complexity of Lie algebra conjugacy in general remains open, but we believe the following is an achievable next target:

\begin{open} What is the complexity of matrix Lie algebra conjugacy for Lie algebras that are \emph{semi}direct products of abelian by semisimple? \end{open}

For the abelian diagonalizable case, our results hold over any field. But for the semisimple and completely reducible cases, we only worked over $\C$. The representation theory of semisimple Lie algebras changes in positive characteristic or over non-algebraically closed fields.

\begin{open}
What is the complexity of Lie algebra conjugacy over algebraically closed fields of positive characteristic? Over $\R$, $\Q$, number fields, or finite fields?
\end{open}

We essentially derandomized Kayal's algorithm for testing equivalence to the determinant, except for a part of the algorithm that is equivalent to polynomial identity testing. It would be nice to know whether there is a way around this, though we suspect there is not:

\begin{open}
Show that testing equivalence to the determinant is as hard as polynomial identity testing, or give a deterministic polynomial-time algorithm for it in the black-box setting. In the dense setting, can equivalence to the determinant be tested in time $poly(t)$ where $t$ is the number of non-zero monomials of the input function?
\end{open}

Finally, there are two avenues for futher progress on the affine equivalence problem using Lie algebra conjugacy.
First, although \lang{GI}-hardness may seem to be the ``final'' word in the short term, in the application to affine equivalence we may be able to avoid \lang{GI} altogether. Lie algebra conjugacy is most directly useful for testing affine equivalence to symmetry-characterized functions such as the determinant. A function $f$ is \definedWord{symmetry-characterized} if for any function $g$, if $g$ has the same symmetries as $f$---that is, $f(A\var{x}) = f(\var{x})$ implies $g(A\var{x}) = g(\var{x})$---then $g$ is a scalar multiple of $f$. Not every Lie algebra can arise as the Lie algebra of the symmetries of a symmetry-characterized function. It is possible that the properties of such Lie algebras are strong enough to avoid graph isomorphism. Second, in addition to the Lie algebra of the symmetries of a function, a function may have a finite group of symmetries ``sitting on top of'' the Lie algebra. 

\begin{open}
What is the complexity of testing conjugacy of finite groups of symmetries, arising from symmetry-characterized functions?
\end{open}

\section*{Acknowledgments}
The author would like to thank the following people for useful discussions regarding this work: Neeraj Kayal, Pascal Koiran, Arakadev Chattopadhyay, J. M. Landsberg, Shrawan Kumar, and Jerzy Weyman. Many of these conversations took place at the Brown-ICERM Workshop on Mathematical Aspects of $\cc{P}$ vs. $\cc{NP}$ and its Variants in August 2011, for which the author would like to thank ICERM and the organizers of the workshop---J. M. Landsberg, Saugata Basu, and J. Maurice Rojas---for the invitation and support to attend the workshop. The author would like to thank Lance Fortnow, Ketan Mulmuley and Benson Farb for their discussions, support, and advice. The author finds it incredibly useful to talk through mathematics with others, and it is his great pleasure to thank Benson Farb, Thomas Church, Ian Shipman, and Jonah Blasiak for not only useful and interesting discussions of this work, but also for their infectious enthusiasm for and injection of fruitful new ideas into this work. In particular, Jonah helped the author clarify his thoughts and together realize the equivalence with graph isomorphism. Finally, this work was partially supported by Ketan Mulmuley's NSF Grant CCF-1017760, Lance Fortnow \etal's NSF Grant DMS-0652521 and fellowships from the U. Chicago Department of Computer Science.

\bibliographystyle{ams-alph}
\bibliography{conj2} 

\appendix

\section{Lie algebra background} \label{sec:background}
For the purposes of this paper, we highly recommend the book of de Graaf \cite{deGraafBook}. We summarize the necessary highlights here. For more general background on Lie algebras we recommend any of several standard books \cite{fultonHarris, jacobson, humphreys, knapp}. Since we are only working over $\C$ for much of this paper, we omit further mention of the field. However, some of the statements and results below hold only if the characteristic of the field is zero, and some only if the field is furthermore algebraically closed.

\subsection{Basic definitions}
A \definedWord{Lie algebra} is a vector space $\subsp{L}$ together with a bilinear operation, referred to as the \definedWord{Lie bracket} and written $[\cdot, \cdot]\colon \subsp{L} \times \subsp{L} \to \subsp{L}$ satisfying:
\begin{enumerate}
\item Skew-symmetry: $[v_{1}, v_{2}] = -[v_{2}, v_{1}]$ (or equivalently, $[v, v] = 0$ for all $v \in \subsp{L}$)

\item Bi-linearity: $[\alpha v + \beta w, u] = \alpha [v, u] + \beta [w, u]$, and similarly for the second coordinate.

\item The Jacobi identity: $[u, [v,w]] + [w, [u, v]] + [v , [w , u]] = 0$. This is the ``Lie algebra'' version of associativity, and can be thought of as ``the derivative of the associative law.''
\end{enumerate} 

A \definedWord{matrix Lie algebra} is a set of matrices where taking $[A, B] := AB - BA$ makes the set into a Lie algebra. In particular, the collection $M_{n}$ of all $n \times n$ matrices is a matrix Lie algebra.

A \definedWord{homomorphism} between Lie algebras $\subsp{L}_{1}, \subsp{L}_{2}$ is a linear map $\rho\colon \subsp{L}_{1} \to \subsp{L}_{2}$ that preserves the brackets, that is, where $\rho([u, v]_{\subsp{L}_{1}}) = [\rho(u), \rho(v)]_{\subsp{L}_{2}}$. An \definedWord{isomorphism} is a bijective homomorphism; an \definedWord{automorphism} is an isomorphism of $\subsp{L}$ with itself.

Note that conjugate matrix Lie algebras are isomorphic as abstract Lie algebras, since $g[M_{1}, M_{2}]g^{-1} = [gM_{1}g^{-1}, gM_{2}g^{-1}]$, that is, conjugation by $g$ is a Lie algebra homomorphism whose inverse is conjugation by $g^{-1}$.

\subsection{Describing Lie algebras as input to algorithms} \label{sec:background:algo}
An abstract Lie algebra is specified in an algorithm by giving a basis for it as a vector space, say $v_{1}, \dotsc, v_{d}$, and its \definedWord{structure constants} $c_{ij}^{(k)}$:
\[
[v_{i}, v_{j}] = \sum_{k=1}^{n} c_{ij}^{(k)} v_{k}.
\] 
Because of the bilinearity of the bracket, the structure constants are enough to determine the value of the bracket on any elements of the Lie algebra: $[\sum \alpha_{i} v_{i}, \sum \beta_{j} v_{j}] = \sum_{ijk} \alpha_{i} \beta_{j} c_{ij}^{(k)} v_{k}$. Each of the axioms of a Lie algebra translates into a condition on the structure constants, for example, skew-symmetry is equivalent to $c_{ij}^{(k)} = -c_{ji}^{(k)}$ for all $i,j,k$.

\subsection{Structure theory of Lie algebras}
Given any two Lie algebras $\subsp{L}_{1}, \subsp{L}_{2}$, their \definedWord{direct sum} is the Lie algebra $\subsp{L}_{1} \oplus \subsp{L}_{2}$ whose underlying vector space is the direct sum of the underlying vector spaces of the $\subsp{L}_{i}$. The bracket $[v_{1}, v_{2}]$ for any elements $v_{1} \in \subsp{L}_{1}$ and $v_{2} \in \subsp{L}_{2}$ is defined to be zero.

An \definedWord{ideal} in a Lie algebra is a subspace $I \subseteq \subsp{L}$ such that $[u,v] \in I$ for any $u \in \subsp{L}$ and $v \in I$. Ideals are the Lie-algebraic analogue of normal subgroups of groups. Given any ideal, one can form the quotient Lie algebra $\subsp{L} / I$ whose elements are additive cosets of $I$, that is, of the form $v + I$; conversely, given any homomorphism of Lie algebras its kernel is an ideal.

A Lie algebra is \definedWord{abelian} if $[u, v] = 0$ for all $u, v \in \subsp{L}$. Any vector space can thus be given the structure of an abelian Lie algebra. Every subspace of an abelian Lie algebra is an ideal.

$0$ is the trivial ideal. An ideal is proper if it is not the whole Lie algebra. In a direct sum $\subsp{L} = \subsp{L}_{1} \oplus \subsp{L}_{2}$, each $\subsp{L}_{i}$ is a proper ideal of $\subsp{L}$. A Lie algebra is \definedWord{simple} if it contains no proper non-trivial ideals, and is non-abelian. (This last condition excludes, for technical reasons, the $1$-dimensional abelian Lie algebra.) A Lie algebra is \definedWord{semisimple} if it is a direct sum of simple Lie algebras.

Over $\C$, the simple Lie algebras have been completely classified for nearly a century. They fall into four infinite families, referred to as type $A_{n}$ ($\mfsl_{n}$, consisting of all trace zero $n \times n$ matrices), $B_{n}$ ($\mathfrak{so}_{2n+1}$, consisting of all $(2n+1) \times (2n+1)$ skew-symmetric matrices $M = -M^{T}$), $C_{n}$ ($\mathfrak{sp}_{2n}$ consisting of all $2n \times 2n$ matrices $M$ satisfying $JM = -M^{T} J$ where $J = \left(\begin{array}{cc} 0 & I_{n} \\
-I_{n} & 0 \end{array}\right)$), and $D_{n}$ ($\mathfrak{so}_{2n}$), and there are five exceptional simple Lie algebras, known as $\mathfrak{e}_{6}$, $\mathfrak{e}_{7}$, $\mathfrak{e}_{8}$, $\mathfrak{f}_{4}$, and $\mathfrak{g}_{2}$.

\subsection{Representations} \label{sec:background:rep}
A \definedWord{representation of a Lie algebra} $\subsp{L}$ is a homomorphism $\rho \colon \subsp{L} \to M_{n}$ for some $n$. A representation is \definedWord{faithful} if this homomorphism is injective. Two representations $\rho_{1}, \rho_{2}\colon \subsp{L} \to M_{n}$ are \definedWord{equivalent} if there is an invertible $n \times n$ matrix $g$ such that $\rho_{1}(v) = g \rho_{2}(v) g^{-1}$ for all $v \in \subsp{L}$.

Equivalence of representations is similar to, but not the same as, conjugacy of matrix Lie algebras. Given two representations $\rho_{1}, \rho_{2} \colon \subsp{L} \to M_{n}$, their images $\subsp{L}_{i} := \im(\rho_{i})$ are matrix Lie algebras. The representations $\rho_{i}$ are equivalent if they are conjugate \emph{as maps}, whereas the matrix Lie algebras forget the maps and only care about their images. In fact, Lemma~\ref{lem:outer} shows that $\subsp{L}_{1}$ and $\subsp{L}_{2}$ are conjugate matrix Lie algebras if and only if $\rho_{1}$ and $\rho_{2}$ are equivalent up to an automorphism of $\subsp{L}$, that is, $\rho_{1}$ is equivalent to $\rho_{2} \circ \alpha$ for some automorphism $\alpha\colon \subsp{L} \to \subsp{L}$. These automorphisms are what cause all the computational difficulties, and allow the equivalences with graph isomorphism and code equivalence.

If $\subsp{L}$ is specified by a basis and structure constants as above, then a representation $\rho\colon \subsp{L} \to M_{n}$ may be specified by giving the $k$ matrices $\rho(v_{i})$ for each basis element.

Given two representations $\rho_{i}\colon \subsp{L} \to M_{n_{i}}$ for $i=1,2$, their \definedWord{direct sum} $\rho_{1} \oplus \rho_{2}\colon \subsp{L} \to M_{n_{1} + n_{2}}$ is defined by the block-matrix:
\[
(\rho_{1} \oplus \rho_{2})(v) = \left(\begin{array}{cc}
\rho_{1}(v) & \\
 & \rho_{2}(v) 
 \end{array}\right).
\]
A representation is called \definedWord{decomposable} if it is (equivalent to) a non-trivial direct sum; otherwise it is called \definedWord{indecomposable}.

The set $M_{n}$ of $n \times n$ matrices acts on the vector space $\F^{n}$ by the usual matrix-vector multiplication. Given a subset $S \subseteq M_{n}$, if $V \subseteq \F^{n}$ is a subspace such that $S \cdot V \subseteq V$, then $V$ is called an \definedWord{$S$-invariant subspace}. The $0$ subspace and the whole space $\F^{n}$ are $S$-invariant for any $S$. 

A representation $\rho\colon \subsp{L} \to M_{n}$ is called \definedWord{irreducible} if $0$ and $\F^{n}$ are the only $\im(\rho)$-invariant subspaces. Otherwise a representation is called \definedWord{reducible}. Note that a decomposable representation is reducible, but the converse need not be true, as illustrated by the example:
\[
\set{\left(\begin{array}{cc}
1 & x \\
  & 1
\end{array}\right) : x \in \F}.
\]

A representation is \definedWord{completely reducible} if it can be decomposed into a direct sum of irreducible representations. Every representation can be decomposed into indecomposable representations; in a completely reducible representation these indecomposables must also be irreducible.

A matrix Lie algebra $\subsp{L} \subseteq M_{n}$, can be viewed as the image of a faithful representation of $\subsp{L}$, namely, take $\rho\colon \subsp{L} \to M_{n}$ to be the inclusion (\ie, identity) map. Via this identification, we also apply the terms (in)decomposable and (ir)reducible to matrix Lie algebras. If $\subsp{L}$ is a completely reducible matrix Lie algebra, then it is equivalent (conjugate) to a matrix Lie algebra consisting of block-diagonal matrices, where the restriction to each block is irreducible.

\begin{thm}[see Theorem~III.10 on p. 81 of Jacobson \cite{jacobson}] \label{thm:CR} A matrix Lie algebra $\subsp{L}$ is completely reducible if and only if $\subsp{L}$ is isomorphic to the direct sum of an abelian, diagonalizable Lie algebra and a semisimple Lie algebra. \end{thm}

The proof of this theorem given in Jacobson \cite{jacobson} is algebraic in nature and can be made effective. All that is required is the solution of a few polynomially sized linear systems of equations. In other words, in polynomial time one can find the irreducible direct summands of a completely reducible representation:

\begin{cor} \label{cor:directSum} Given a completely reducible matrix Lie algebra $\subsp{L} \subseteq M_{n}$, one can find in $poly(n)$ time a matrix $g$ so that $g\subsp{L}g^{-1}$ is the direct sum of an abelian diagonal Lie algebra and a semisimple Lie algebra, where the semisimple part consists of block-diagonal matrices, each block being irreducible. \end{cor}

\subsection{Inner and Outer Automorphisms} \label{sec:background:aut}
The collection of automorphisms of a Lie algebra $\subsp{L}$ form a group $\Aut(\subsp{L})$ under composition of maps. Given a Lie algebra $\subsp{L}$ and $v \in \subsp{L}$, the Jacobi identity implies that the map $\ad_{v}\colon \subsp{L} \to \subsp{L}$ defined by $\ad_{v}(u) := [v, u]$ is a homomorphism of Lie algebras. If $\ad_{v}^{k} := \ad_{v} \circ \dotsb \circ \ad_{v}$ is the zero map for $k$ sufficiently large, then $\exp(\ad_{v}) := I + \ad_{v} + \frac{1}{2} \ad_{v}^{2} + \dotsb + \frac{1}{(k-1)!} \ad_{v}^{k-1}$ is an automorphism of $\subsp{L}$. Automorphisms arising in this way are called \definedWord{inner automorphisms}. The inner automorphisms form a normal subgroup $\Inn(\subsp{L}) \leq \Aut(\subsp{L})$. The quotient group $\Aut(\subsp{L}) / \Inn(\subsp{L})$ is called the \definedWord{outer automorphism} group and is denoted $\Out(\subsp{L})$.

The outer automorphism groups of the simple Lie algebras are completely known:
\[
\begin{array}{rclcrcl}
\Out(\mfsl_{n}) & = & S_{2} & \qquad & \Out(\mathfrak{sp}_{2n}) & = & 1 \\
(n \neq 4) \quad \Out(\mathfrak{so}_{2n}) & = & S_{2} & \qquad & \Out(\mathfrak{so}_{2n+1}) & = & 1 \\
\Out(\mathfrak{so}_{8}) & = & S_{3} & \qquad & \Out(\mathfrak{e}_{7}) & = & 1 \\
\Out(\mathfrak{e}_{6}) & = & S_{2} & \qquad & \Out(\mathfrak{e}_{8}) & = & 1 \\
 & & & & \Out(\mathfrak{f}_{4}) & = & 1 \\
 & & & & \Out(\mathfrak{g}_{2}) & = & 1 \\
\end{array}
\]
The action of $\Out(\mfsl_{n})$ on the representations of $\mfsl_{n}$ is trivial. The action technically takes a representation to its dual, but for $\mfsl_{n}$, the dual of a representation is equivalent to that representation.

\subsection{Twisting representations by automorphisms}
Given an automorphism $\alpha \colon \subsp{L} \to \subsp{L}$ and a representation $\rho\colon \subsp{L} \to M_{n}$, we get another representation $\rho \circ \alpha \colon \subsp{L} \stackrel{\alpha}{\to} \subsp{L} \stackrel{\rho}{\to} M_{n}$, given by $(\rho \circ \alpha)(v) = \rho(\alpha(v))$. Since $\alpha$ is an automorphism, it is, in particular, onto, so $\im(\rho \circ \alpha) = \im(\rho)$. However, $\rho \circ \alpha$ and $\rho$ need not be equivalent as representations, despite having the same image. We call $\rho \circ \alpha$ the \definedWord{twist} of the representation $\rho$ by the automorphism $\alpha$.

For semisimple Lie algebras, twisting by inner automorphisms does in fact lead to equivalent representations:

\begin{lem}[see Lemma~8.5.1 in de Graaf \cite{deGraafBook}] \label{lem:inner} Let $\rho\colon \subsp{L} \to M_{n}$ be a representation of a semisimple Lie algebra $\subsp{L}$ and let $\alpha$ be an inner automorphism of $\subsp{L}$. Then $\rho \circ \alpha$ is equivalent to $\rho$.
\end{lem}

Since twisting a representation by an inner automorphism sends it to an equivalent representation, we find that the outer automorphism group $\Out(\subsp{L})$ acts on the set of representations-up-to-equivalence. If $\alpha \in \Out(\subsp{L})$, we denote the image of $\rho$ under the action of $\alpha$ by $\rho^{\alpha}$. Equivalently, let $\alpha_{*} \in \Aut(\subsp{L})$ be a representative of $\alpha \in \Out(\subsp{L})$; then $\rho^{\alpha}$ is the equivalence class of $\rho \circ \alpha_{*}$, and by the lemma, this equivalence class is independent of the choice of representative $\alpha_{*}$.

We note that the same result is vacuously true for abelian Lie algebras, since if $\subsp{L}$ is abelian then it has no non-trivial inner automorphisms. Hence it also holds for Lie algebras that are a direct sum of abelian and semisimple.

\subsection{More structure theory} \label{sec:background:morestruct}
Given two ideals $A,  B \subseteq \subsp{L}$, their commutator is defined as $[A, B] := \Span\set{[a, b] : a \in A, b \in B}$; the commutator of two ideals is again an ideal (this is an exercise in the Jacobi identity). The \definedWord{derived series} of $\subsp{L}$ is defined as follows: $\subsp{L}^{(0)} := \subsp{L}$, $\subsp{L}^{(i+1)} := [\subsp{L}^{(i)}, \subsp{L}^{(i)}]$. $\subsp{L}^{(1)} = [\subsp{L}, \subsp{L}]$ is called the \definedWord{derived or commutator subalgebra}.

\begin{defn} A Lie algebra $\subsp{L}$ is \definedWord{solvable} if the derived series terminates at $\subsp{L}^{(k)} = 0$ for some $k$. \end{defn}

Each step in the derived series, $\subsp{L}^{(i)} / \subsp{L}^{(i+1)}$ is abelian, so solvable Lie algebras are ``iterated extensions of abelian Lie algebras.''

The \definedWord{lower central series} is defined by $\subsp{L}_{0} := \subsp{L}$ and $\subsp{L}_{i+1} := [\subsp{L}, \subsp{L}_{i}]$. Note that here we take the commutator of $\subsp{L}_{i}$ with the whole of $\subsp{L}$, rather than just with $\subsp{L}_{i}$ (as in the derived series). Hence the lower central series decreases more slowly than the derived series.

\begin{defn} A Lie algebra $\subsp{L}$ is \definedWord{nilpotent} if the lower central series terminates at $\subsp{L}_{k} = 0$ for some $k$. \end{defn}

Finally, in order to state the main structural theorems of Lie algebras, we define semidirect products and derivations. A \definedWord{derivation} on a Lie algebra $\subsp{L}$ is a linear map $d\colon \subsp{L} \to \subsp{L}$ such that $d([u, v]) = [u, d(v)] + [d(u), v]$. Note the similarity with the product rule for differentiation. Since a derivation is a linear map, we may compose two derivations as linear maps; then defining $[d_{1}, d_{2}] := d_{1} \circ d_{2} - d_{2} \circ d_{1}$ makes the collection of derivations of $\subsp{L}$ into a Lie algebra denoted $\Der(\subsp{L})$. 

Given two Lie algebras $\subsp{L}_{1}, \subsp{L}_{2}$ and a homomorphism $\varphi \colon \subsp{L}_{2} \to \Der(\subsp{L}_{1})$, we define the semi-direct product $\subsp{L}_{1} \rtimes_{\varphi} \subsp{L}_{2}$ as follows. The underlying vector space is the direct sum of $\subsp{L}_{1}$ and $\subsp{L}_{2}$. On each of these subspaces, the Lie bracket is defined as it was originally. If $v \in \subsp{L}_{1}$ and $d \in \subsp{L}_{2}$ we define
\[
[v, d] := d(v).
\]
Extending by linearity and skew-symmetry, we find 
\[
[v_{1} + d_{1}, v_{2} + d_{2}] = [v_{1}, v_{2}] + d_{2}(v_{1}) - d_{1}(v_{2}) + [d_{1}, d_{2}]
\]
where $v_{i} \in \subsp{L}_{1}$ and $d_{i} \in \subsp{L}_{2}$.

The following two theorems are quite strong structural theorems. For example, nothing even close to these holds in the case of finite groups, despite the similarity in the definitions of all the notions (nilpotent, solvable, semidirect product).

\begin{thm}[Levi's Theorem, cf. \S III.9, p. 91 of Jacobson \cite{jacobson}] \label{thm:levi} Every Lie algebra is the semidirect product of a solvable Lie algebra by a semisimple one. (That is, the semisimple one acts as derivations on the solvable one.) \end{thm}

\begin{thm}[see Corollary~II.7.1 on p. 51 of Jacobson \cite{jacobson}] A Lie algebra is solvable if and only if its derived subalgebra is nilpotent.
\end{thm}

\begin{rmk} \label{rmk:buildingBlocks} Since solvable Lie algebras are iterated extensions of abelian ones (see above), and considering Theorem~\ref{thm:levi}, we may say that abelian and simple Lie algebras form the ``building blocks'' of all Lie algebras. \end{rmk}


\section{Reduction from \mylang\ to graph isomorphism} \label{app:aGI}
\begin{proof}[Proof of Lemma~\ref{lem:aGI}]
First, if the permutation groups $G_{\ell}$ are all trivial, then we can take each $M_{i}$ as the bipartite adjacency matrix of a vertex-colored and edge-colored bipartite graph. The vertices corresponding to the columns are colored according to their part in the column partition; we refer to these vertices as column-vertices. The edges are colored by the integer entries of each $M_{i}$. It is clear that the $M_{i}$ are equivalent if and only if the corresponding vertex- and edge-colored bipartite graphs are bipartite-color-isomorphic, that is, isomorphic by an isomorphism which preserves the two parts of the bipartition and preserves each color class of vertices and each color class of edges.

To handle the permutation groups $G_{i}$ we make one additional step in the reduction. Since there is one $G_{i}$ for each column $i$, we must encode its action on the edge-labels incident to the column-vertex $i$. To do this we add a ``color palette'' gadget for each column-vertex, which will encode both the edge-labels, as well as enforcing the action of $G_{i}$ on these labels. That is, the color palette will be such that the way automorphisms of the resulting graph act on the encoding of the edge-labels is exactly the same as $G_{i}$ acts on them.

To encode the edge-labels with the color palette, we divide each edge by a new vertex, and attach this new vertex to the vertex of the color palette which encodes the appropriate edge color. We only need color palettes capable of encoding permutation actions of the trivial group, $S_{2}$, and $S_{3}$.

\emph{$G_{i}$ trivial.} If some $G_{i}$ is trivial, the corresponding color palette is simply a line of vertices with a marked vertex at the end. The marked vertex prevents reversing the order of the line, and the different vertices in the line encode the different edge labels on the edges incident to column-vertex $i$.

\emph{$G_{i} \cong S_{2}$.} $S_{2}$ has two possible orbit types (=transitive actions): a single fixed point, or an orbit of size two. The color palette is the disjoint union of two graphs corresponding to the two possible orbit types. Each of these graphs has its own marked vertex at the end. One of these two graphs is simply a line as in the previous case: the vertices of this line correspond to those edge-labels that are fixed by the action of $S_{2}$. The other graph is the disjoint union of two lines, each of which is joined at the end to the marked vertex. The action of $S_{2}$ swaps the $i$-th vertex of one of these lines with the $i$-th vertex of the other. This enforces that the action of the edge group $G_{i}$ either swaps all of the edge labels (that is, via the nontrivial element of $S_{2}$) or none of them.

For the sake of the next case, it is useful to think of this color palette as gluing together in a line multiple copies of the ``color gadget'' consisting of two disconnected vertices.

\emph{$G_{i} \cong S_{3}$.} $S_{3}$ has four orbit types: 1) the trivial action, 2) the action on two points by which odd permutations swap the points and even permutations fix them, 3) the natural action of $S_{3}$ on three points, and 4) the regular action of $S_{3}$ on itself ($6$ points). However, these last three orbit types must be linked, since if an element of $S_{3}$ swaps two points according to (2), it must also have some action according to (3) and (4). Thus the color palette in this case is the disjoint union of two palettes: the trivial, line palette as before, and a more complicated palette encoding the actions (2)--(4).

This more complicated palette is given by a ``color gadget,'' multiple copies of which are glued together in a line, as in all the other cases. The color gadget is as follows:

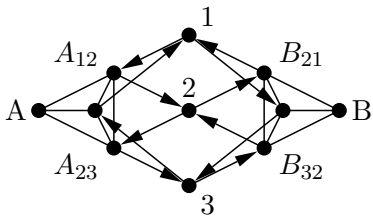
\begin{figure}[h]
\centering 
\begin{graph}(4,4)(-2,-2)
\roundnode{1}(0.00, 1.00)
\freetext(0.25, 1.25){1}
\roundnode{2}(0.00, 0.00)
\freetext(0.00, 0.30){2}
\roundnode{3}(0.00, -1.00)
\freetext(0.25, -1.25){3}
\roundnode{A12}(-1.00, 0.50)
\freetext(-1.50, 0.75){$A_{12}$}
\roundnode{A23}(-1.00, -0.50)
\freetext(-1.50, -0.75){$A_{23}$}
\roundnode{A31}(-1.25, 0.00)
\roundnode{A}(-2.00, 0.00)
\freetext(-2.30, 0.00){A}
\roundnode{B21}(1.00, 0.50)
\freetext(1.50, 0.75){$B_{21}$}
\roundnode{B32}(1.00, -0.50)
\freetext(1.50, -0.75){$B_{32}$}
\roundnode{B13}(1.25, 0.00)
\roundnode{B}(2.00, 0.00)
\freetext(2.30, 0.00){B}

\edge{A12}{A23}
\edge{A23}{A31}
\edge{A31}{A12}

\edge{A}{A12}
\edge{A}{A23}
\edge{A}{A31}

\edge{B21}{B32}
\edge{B32}{B13}
\edge{B13}{B21}

\edge{B}{B21}
\edge{B}{B32}
\edge{B}{B13}

\diredge{1}{A12}
\diredge{A12}{2}
\diredge{2}{A23}
\diredge{A23}{3}
\diredge{3}{A31}
\diredge{A31}{1}

\diredge{1}{B13}
\diredge{B13}{3}
\diredge{3}{B32}
\diredge{B32}{2}
\diredge{2}{B21}
\diredge{B21}{1}

\end{graph}
\end{figure}
 
Multiple copies of this color gadget are glued together along three lines, one connecting the ``$1$'' vertices, one connecting the ``$2$'' vertices, and one connecting the ``$3$'' vertices. At one end of these lines, every vertex in the color gadget is connected to a new marked vertex, to prevent the line from being swapped end-to-end.

A set of edge colors corresponding to an orbit of type (2) is encoded by the $A$ and $B$ vertices.

A set of edge colors correspondgin to an orbit of type (3) is encoded by the vertices $1$, $2$, and $3$. 

A set of edge colors corresponding to an orbit of type (4) is encoded by the vertices $A_{12}$, $A_{23}$, $A_{31}$, $B_{21}$, $B_{32}$, and $B_{13}$. $A_{31}$ and $B_{13}$ are not labelled in the diagram due to space, but there are directed edges $3 \to A_{31} \to 1$ and $1 \to B_{13} \to 3$. 

It remains to show that this color gadget really works as desired. Let us examine the automorphisms of the color gadget. We claim that the automorphism group is $S_{3}$, that it acts on the vertices $1$, $2$, $3$ in its natural action (3), it acts on the $A_{ij}$'s and $B_{ji}$'s together in its regular action (4), and it acts on $A, B$ in its odd-even action (2). 

$1$, $2$, and $3$ are the only vertices with in-degree and out-degree $1$, so at most they can be swapped amongst each other. Hence the automorphism group is at most $S_{3}$. To show the above claim, it suffices to show that the generating set $(123)$ and $(12)$ of $S_{3}$ provides automorphisms of the color gadget that act as described.

Consider first $(123)$. It acts on $1$, $2$, and $3$ as described by the cycle notation: $1 \mapsto 2 \mapsto 3 \mapsto 1$. To be an automorphism, it is then forced to send $A_{12} \mapsto A_{23} \mapsto A_{31} \mapsto A_{12}$ and $B_{21} \mapsto B_{32} \mapsto B_{13} \mapsto B_{21}$. Note that $(123)$ cannot possibly swap the $A_{ij}$'s and the $B_{ji}$'s, since the directed edges determine an orientation that is preserved by $(123)$. Moreover, its action on these vertices is exactly the action of $(123)$ by right multiplication on $S_{3}$ itself. This implies that $(123)$, and hence all the even permutations, fix the vertices $A$ and $B$.

Next, consider $(12)$. Since $(12)$ reverses the orientation determined by the directed edges, it must swaps the $A_{ij}$'s and $B_{ji}$'s, as follows: $A_{12} \leftrightarrow B_{21}$, $A_{23} \leftrightarrow B_{13}$, and $A_{31} \leftrightarrow B_{32}$. This also implies that $A \leftrightarrow B$. Hence odd permutations swap $A$ and $B$. Finally, the action of $(12)$ on the $A_{ij}$'s and $B_{ji}$'s is in accordance with the right regular action of $(12)$ on $S_{3}$, compatible with that of $(123)$ above. We can put this together through the correpsondence:
\begin{eqnarray*}
() & \sim & A_{12} \\
(123) & \sim & A_{23} \\
(132) & \sim & A_{31} \\
(12) & \sim & B_{21} \\
(13) & \sim & B_{32} \\
(23) & \sim & B_{13}
\end{eqnarray*}
\end{proof}

We suspect that these ideas can be extended to show that the general twisted code equivalence problem, as defined in Codenotti \cite{codenottiPhd} Karp-reduces to graph isomorphism.

\end{document}